\documentclass{iopart}

\usepackage{iopams}
\usepackage{setstack}
\usepackage{amsthm}
\usepackage{braket}
\usepackage{graphicx}% Include figure files
\usepackage{tikz-cd}
\usepackage{dcolumn}% Align table columns on decimal point
\usepackage{bm}% bold math
\usepackage{epsf}

\def\polhk#1{\setbox0=\hbox{#1}{\ooalign{\hidewidth
      \lower1.5ex\hbox{`}\hidewidth\crcr\unhbox0}}}

\newtheorem{theorem}{Theorem}%[section]
\newtheorem{lemma}[theorem]{Lemma}
\newtheorem{proposition}[theorem]{Proposition}
\newtheorem{corollary}[theorem]{Corollary}

\begin{document}
\title[Upper Bounds for the Recurrence Time]{Upper Bounds for the Poincar\'{e} Recurrence Time in Quantum Mixed States}
%    Information for first author
\author{V Gimeno$^1$ and J M Sotoca$^2$}
%    Address of record for the research reported here
\address{$^1$ Departament de Matem\`{a}tiques- Institut Universitari de Matem\`{a}tiques i Aplicacions, Universitat Jaume I, Castell\'o,
Spain.}
%    Current address
%\curraddr{}
\ead{gimenov@uji.es}
%    \thanks will become a 1st page footnote.
% \thanks{The first author was supported in part by NSF Grant \#000000.}

\address{$^2$ Departamento de Lenguajes y Sistemas Inform\'aticos- Institute of New Imaging Technologies, Universitat Jaume I, Castell\'o,
Spain.%\\This line break forced with \textbackslash\textbackslash
}%
 \ead{sotoca@uji.es}

\begin{abstract}In this paper by using geometric techniques,  we provide upper bounds for the Poincar\'e recurrence time of a quantum mixed state with discrete spectrum of energies. In the case of discrete but finite spectrum we obtain two type of upper bounds; one of them depends on the uncertainty in the energy, and the other depends only on the (finite) number of states. In the case of discrete but non-finite spectrum we obtain in the same way two upper bounds defining the number of relevant states according to an statistical measurement. These bounds correspond to two different situations in the quantum recurrence process. The first bound is a recurrence time estimation purely quantum, while the other bound that is related with the number of relevant states survives in the classical limit.
\end{abstract}

%    General info
\pacs{03.65.-w, 03.67.Hk}
% Uncomment for keywords
\vspace{2pc}
\noindent{\it Keywords}: Poincar\'e recurrence time, mixed states, fidelity

%\keywords{}

%\submitto{\jpa}

\maketitle

%\dedicatory{This paper is dedicated to [[[]]].}

\maketitle
%\setcounter{tocdepth}{1}
%\tableofcontents

\section{Introduction}
The classical Poincar\'e recurrence theorem states that an isolated mechanical system with a fixed finite energy and in a fixed bounded volume, will return after a long enough time, close to its initial mechanical state. Poincar\'e recurrence theorem follows from Liouville's theorem (see \cite{Arnold} for instance) due to the volume-preserving property of the Hamiltonian flux of the classical phase space.
Nevertheless, the total amount of the volume of the phase space, and hence the recurrence time (length of time elapsed until the recurrence), depends on the Hamiltonian of the system.
%tret de  Classical Theory Per Paul McEvoy.

\begin{figure}[h]\label{fig1}
\begin{center}
\includegraphics[scale=0.20]{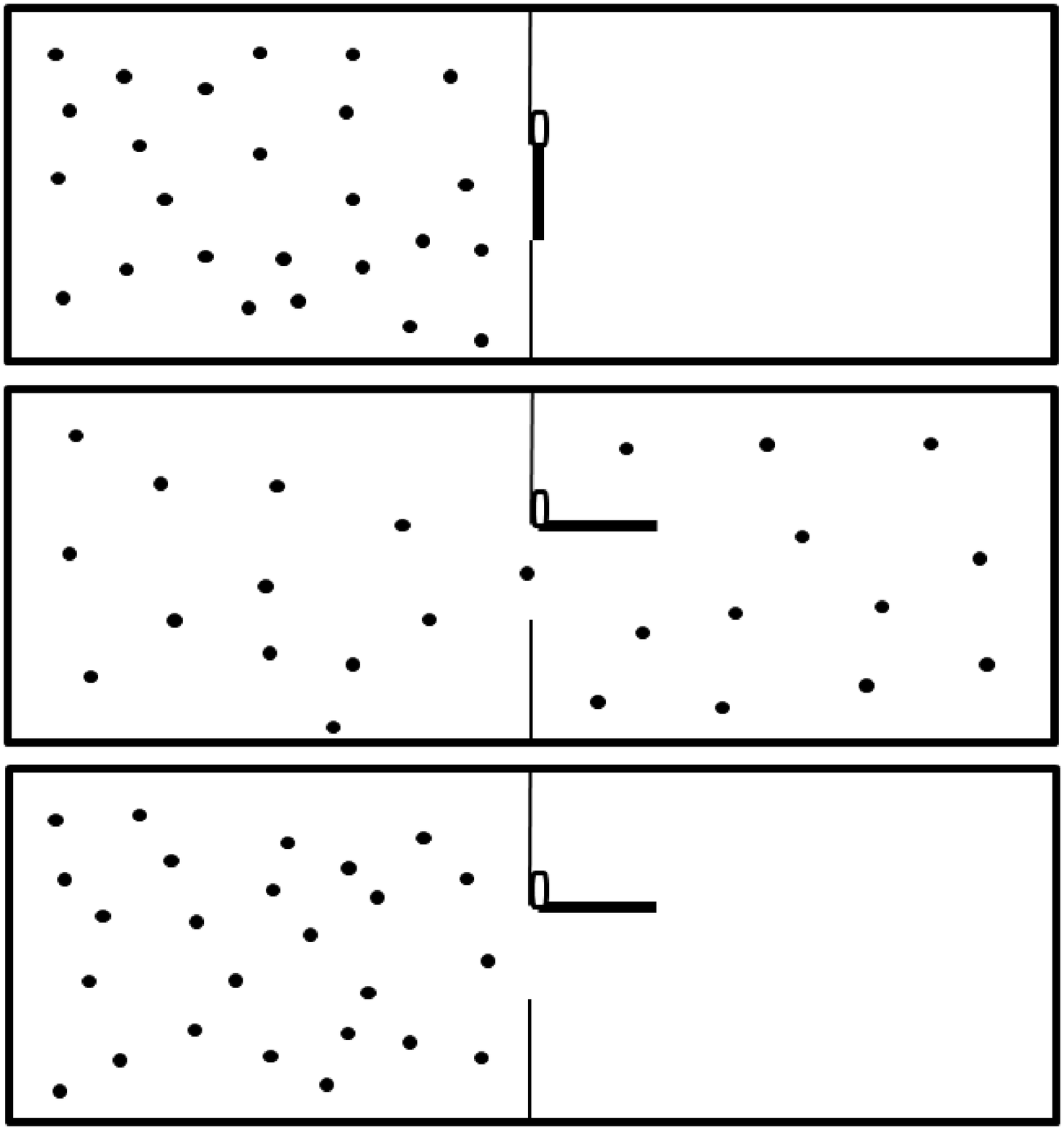}
\end{center}
\caption{The ``recurrence paradox.'' If we let a gas expand in a bounded recipient after a sufficient large time it will come back very close to the initial position.}
\end{figure}

The Poincar\'e recurrence theorem has counterintuitive implications when is considered within the context of the Second Law of Thermodynamics. According to this law, the measure of disorder of a system will never decrease - it will either increase or stay the same. For example, considering an isolated system - if you open a partition separating a chamber that contain gas and a chamber in vacuum, after a time, the gas molecules will again be collected in the first chamber (see Figure\ref{fig1}). This is known as the recurrence paradox and is most commonly reconciled by a claim that the amount of time that one must wait before the gas system returns to its initial state is orders of magnitude larger than the expected life of the universe.
% Tret de Applying the Poincar ́e Recurrence Theorem to Billiards IMPORTANT: Poincare Recurrence, Zermelo’s Second Law Paradox.

Indeed Zermelo used Poincar\'e result's to argue that the Boltmann's formula for the entropy should decrease after some sufficiently long time. Poincar\'e argument does not give any indication on how to estimate such recurrence times and it is an important question from a theoretical point of view.

In the quantum world there is an equivalence of the Poincar\'e recurrence principle (see \cite{Bocchieri,Schulman,Percival,Hogg,Chandrashekar2010,Venuti2015}). Thus, these quantum systems can be considered having time-periodic Hamiltonians in which the quasi energy spectrum is discrete. Furthermore, the evolution of these systems is almost-periodic even though the system may be in a non-stationary state or in a mixed state.

The Poincar\'e recurrence is relevant in order to understand ``non-reversible'' phenomena, such as the decoherence of a quantum system induced by the environment (see \cite{Venugopalan,Castagnino,DePonte,Berman,Zurek1982}) where in order to obtain effective decoherence, a small quotient between the decoherence and recurrence time is required. An other place where the Poincar\'e recurrence could play an important role is the lost of information in quantum black-holes (see \cite{Dyson,Barbon,Siopsis}).

Consider the evolution of an initial pure state $\ket{\Psi_0}\in\mathcal{H}$ of the Hilbert space $\mathcal{H}$ by the Hamiltonian operator $H$ given by the following Schr\"odinger's equation

\begin{equation}\label{sch}
\frac{\partial}{\partial t}\ket{\Psi(t)}=-\frac{i}{\hbar}H\ket{\Psi(t)},\quad \ket{\Psi(0)}=\ket{\Psi_0}.
\end{equation}
Recurrence implies in particular that for any $\epsilon>0$, there exists $t_0$ large enough such that
$$
\Vert \ket{\Psi(t_0)}-\ket{\Psi_0}\Vert^2<\epsilon
$$

Under certain type of wave packet (see \cite{Venugopalan,Castagnino,DePonte,Berman,Zurek1982,Kaminishi2015}) estimations for the recurrence time can be obtained. In \cite{Peres}, assuming that $\ket{\Psi_0}$ has a finite decomposition in the basis of eigenfunctions $\{\ket{i}\}$ of the Hamiltonian $H$, \emph{i.e.},
\begin{equation}\label{decomp0}
\ket{\Psi_0}=\sum_{i=1}^nz_i\ket{i}, \quad H\ket{i}=2\pi\hbar\nu_i\ket{i},\quad \Vert\ket{\Psi_0}\Vert=\Vert\ket{i}\Vert=1,
\end{equation}
 the following estimation of the recurrence time appears
\begin{equation}\label{estimation1}
t_{\rm rec}\sim\frac{1}{n^{1/2}\nu\sigma}
\end{equation}
where
$$
\nu=\frac{1}{n}\sum_{i=1}^n\nu_i,\quad\sigma=\frac{\pi^{\frac{n-1}{2}}R^{n-1}}{\Gamma\left[(n+1)/2\right]},\quad R=\frac{(n\epsilon)^{1/2}}{2\pi}
$$
In \cite{Bhartacharyya} the authors obtain several expressions for the recurrence time, between them we remark  the following estimation
\begin{equation}\label{estimation2}
t_{\rm rec}\sim\frac{1}{(n-1)^{1/2}\overline{\nu}_{m1}}\Gamma\left(\frac{n}{2}\right)\left(\frac{8\pi}{\epsilon\left(n-1\right)}\right)^{(n-2)/2}
\end{equation}
where here
$$
\overline{\nu}_{m1}=\frac{1}{\sqrt{n-1}}\left(\sum_{m=2}^n\left(\nu_{m}-\nu_1\right)^2\right)^{\frac{1}{2}}
$$

%\subsection{Pure states: Discrete and finite spectrum}

We should remark here that the estimations given in equations (\ref{estimation1}) and (\ref{estimation2}) depend on the average of the energy (or the gap energy) of the quantum system.

%\subsection{Mixed states with discrete spectrum}

The general case in quantum mechanics deals with mixed states. A mixed state cannot be described as a ket vector. Instead, it is described by its associated density matrix (or density operator). The temporal evolution of a mixed state $\rho_0$ is given by the von-Newmann law
$$
\dot\rho(t)=-\frac{i}{\hbar}[H,\rho(t)],\quad \rho(0)=\rho_0.
$$
The transition probability between the mixed state $\rho$ to the mixed state $\sigma$ is given now by the  \emph{fidelity}  $F(\rho,\sigma)$
$$
F\left(\rho,\sigma\right):= \tr\left[\sqrt{\sqrt{\rho}\sigma\sqrt{\rho}}\right].
$$

Surprising enough in the case of finite number of states there is an upper bound for the recurrence time that does not depend on the specific Hamiltonian. In fact we can state

\begin{theorem}\label{mixed}
Let $\rho_0$ be a mixed state of the Hilbert space $\mathcal{H}$ of finite dimension $n={\rm dim}(\mathcal{H})$. Let $\rho(t)$ denote the unitary evolution given by the Hamiltonian $H$, \emph{i.e.},
$$
\dot\rho(t)=-\frac{i}{\hbar}[H,\rho(t)],\quad \rho(0)=\rho_0.
$$
Then, for any $t>0$ and any $\epsilon\in (0,1]$ there exists a time $t_{\rm rec}$ such that
\begin{equation}\label{eq:9.5}
F\left(\rho_0,\rho\left(t_{\rm rec}\right)\right) \geq \epsilon,
\end{equation}
with $t_{\rm rec}=j\cdot t$, $j\in \mathbb{N}$ and such that
$$
1\leq j\leq \sqrt{\pi}\frac{\Gamma(n^2)}{\Gamma(n^2+\frac{1}{2})}\frac{1}{\int_0^{\frac{\sqrt{2-2\epsilon}}{2}}\sin^{2n^2-2}(s)ds}\quad\cdot
$$
\end{theorem}

Inequality (\ref{eq:9.5}) implies by the the Fuchs–van de Graaf inequalities
\begin{equation}\label{eq:10.5}
\Vert \rho(t_{\rm rec}(\epsilon))-\rho_0\Vert^2\leq 4\left(1-\epsilon^2\right).
\end{equation}

If we take into account the uncertainty in the energy we can recover for a mixed state equivalent expressions  to (\ref{estimation1}) and (\ref{estimation2}). The statement of the following theorem in fact do so

\begin{theorem}\label{thm:1.3.2}
Let $\rho_0$ be a mixed state of the Hilbert space $\mathcal{H}$ of finite dimension $n={\rm dim}(\mathcal{H})$. Let $\rho(t)$ denote the unitary evolution given by the Hamiltonian $H$, \emph{i.e.},
$$
\dot\rho(t)=-\frac{i}{\hbar}[H,\rho(t)],\quad \rho(0)=\rho_0.
$$
Suppose that the initial mixed state $\rho_0$ has non-zero uncertainty in the energy $\triangle E_{\rho_0}\neq 0$ (with $\triangle E_{\rho_0}=\sqrt{\tr\left(H^2\rho_0\right)-\tr\left(H\rho_0\right)^2}$). Then, for any $\epsilon>0$ with
$$
\epsilon< \pi\cdot\min_{k\in\{1,\cdots,n\}}\left\{\sqrt{\tr\left(\rho_0 \ket{k}\bra{k}\right)}
\right\}
$$where $\{\ket{k}\}$ is a basis of ortonormal eigenstates of the Hamiltonian ( \emph{i.e.}, $H\ket{k}=E_k\ket{k}$ and $\braket{j\vert k}=\delta_{j,k}$), there exists a time $t_{\rm rec}(\epsilon)$ such that
\begin{equation}\label{eq:5.2}
F\left(\rho_0,\rho\left(t_{\rm rec}(\epsilon)\right)\right) \geq 1-\frac{\epsilon^2}{4},
\end{equation}
with %$\triangle$ $\left( \right)$ $\underset{k=1}\Pi$
\begin{equation}\label{eq:5.5}
\epsilon\frac{\hbar}{\triangle E_{\rho_0}}\leq t_{\rm rec}(\epsilon)\leq \left(\frac{c_n\overset{n}{\underset{k=1}{\Pi}}\sqrt{\tr\left(\rho_0 \ket{k}\bra{k}\right)}}{\epsilon^{n-1}
}\right)\,\frac{\hbar}{\triangle E_{\rho_0}}
\end{equation}
where $c_n=(n-1)\Gamma(\frac{n-1}{2})4^{n-1}\pi^{\frac{n+1}{2}}$.

\end{theorem}

The estimation (\ref{eq:5.5}) can be simplified by using the mean-arithmetic mean-geometric inequality
\begin{eqnarray*}
  \overset{n}{\underset{k=1}{\Pi}}\sqrt{\tr\left(\rho_0 \ket{k}\bra{k}\right)}&=&\left[\left(\overset{n}{\underset{k=1}{\Pi}}\tr\left(\rho_0 \ket{k}\bra{k}\right)\right)^\frac{1}{n}\right]^\frac{n}{2}\\
  &\leq& \left[\frac{\overset{n}{\underset{k=1}{\sum}}\tr\left(\rho_0 \ket{k}\bra{k}\right)}{n}\right]^\frac{n}{2}=\frac{1}{n^{\frac{n}{2}}}
\end{eqnarray*}

Hence under the same hypothesis of the above theorem we can state
\begin{equation}
\epsilon\frac{\hbar}{\triangle E_{\rho_0}}\leq t_{\rm rec}(\epsilon)\leq  \left(\frac{c_n}{\epsilon^{n-1}n^{\frac{n}{2}}}\right)\frac{\hbar}{\triangle E_{\rho_0}}
\end{equation}
Note that the lower bound is the Mandelstam-Tamm inequality (see \cite{Mandelstam}). Observe, moreover, that by using the Fuchs-van de Graaf inequalities, the inequality (\ref{eq:5.2}) implies
$$
1-\frac{\epsilon^2}{4}\leq\sqrt{1-\frac{1}{4}\Vert\rho(t_{\rm rec}(\epsilon))-\rho_0 \Vert^2}
$$
Therefore,
\begin{equation}\label{eq:8.5}
\Vert\rho(t_{\rm rec}(\epsilon))-\rho_0 \Vert^2\leq 4\left(1-\left(1-\frac{\epsilon^2}{4}\right)^2\right)=2\epsilon^2\left(1-\frac{\epsilon^2}{8}\right)
\end{equation}

We can use this inequality to obtain upper bounds in the case of a discrete but non-finite spectrum, from the problem statement proposed in \cite{Percival} where the density matrix is an almost periodic function of time. Consider $\rho(t)$ the density matrix of a system for a set of discrete stationary states, with energy levels $E_k, k=0,1,2,\cdots,$ where some of them may have the same value if they are degenerate. In energy representation the matrix elements are defined as
\begin{equation*}
\rho_{kk'}(t)=\langle k\vert \rho(t)\vert k'\rangle\quad.
\end{equation*}

Let $T_k=\vert k\rangle\langle k\vert$ be the projection operator onto the $k$th stationary state, then
\begin{equation*}
\rho^{kk'}(t)=T_k\rho(t)T_{k'}\quad,
\end{equation*}
is the matrix which energy representation has only one nonzero element, equal to $\rho_{kk'}(t)$ and located in $(k,k')$. These matrices are orthogonal in density space
\begin{equation*}
\tr\left(\rho^{kk'}(t)^\dag\rho^{k''k'''}(t)\right)=\delta_{kk''}\delta_{k'k'''}\vert\rho_{kk'}(t)\vert^2\quad,
\end{equation*}
and
\begin{eqnarray*}
\rho(t)&=\sum_{k=0}^{\infty}\sum_{k'=0}^{\infty}\rho^{kk'}(t)\\
&=\sum_{k=0}^{\infty}\sum_{k'=0}^{\infty}\rho^{kk'}(0)e^{i\omega_{kk'}t}\quad,
\end{eqnarray*}
where $\omega_{kk'}=\frac{E_{k'}-E_k}{\hbar}$.
Now, regard the finite sum
\begin{equation*}
\sigma^{N}(t)=\sum_{k=0}^N \sum_{k'=0}^{N}\rho^{kk'}(t)\quad,
\end{equation*}
as an approximation to $\rho(t)$. Then, the squared error is
\begin{eqnarray*}
\Vert \rho(t)-\sigma^{N}(t)\Vert^2&=\Vert \sum_{n=N+1}^\infty \sum_{k'=N+1}^{\infty}\rho^{kk'}(t)\Vert^2\\
&=\sum_{k=N+1}^\infty \sum_{k'=N+1}^{\infty}\Vert \rho^{kk'}(t)\Vert^2\\
&=\sum_{k=N+1}^\infty \sum_{k'=N+1}^{\infty}\Vert \rho^{kk'}(0)\Vert^2\quad.
\end{eqnarray*}

The second equality is obtained from the orthogonality of $\rho^{kk'}$. Because the error is not time-dependent, $\sigma^{N}(t)$ converges uniformly to $\rho(t)$ (in the $\Vert\, \Vert$-norm sense).
Let us denote by $\delta_N$ the time-independent quantity
$$
\delta_N:=\Vert \rho(t)-\sigma^{N}(t)\Vert^2=\Vert \rho(0)-\sigma^{N}(0)\Vert^2
$$
So, $\rho(t)$ can be approximated by $\sigma^{N}(t)$ in the sense that $\delta_N\to 0$ when $N\to \infty$.
In such a case  we shall say that the mixed state $\rho$ has $N$ \emph{relevant states with error term}  of  $\delta_N$.

Observe that the trace of $\sigma^{N}(t)$ is time-independent because
$$
\tr(\sigma^{N}(t))=\tr(e^{\frac{-iH}{\hbar}t}\sigma^{N}(0)e^{\frac{iH}{\hbar}t})=\tr(\sigma^{N}(0))
$$
We can define
$$
\widetilde\sigma^{N}(t):
=\frac{1}{\tr(\sigma^{N}(0))}\sigma^{N}(t)
$$
then
$$
\widetilde\sigma^{N}(t)=e^{\frac{-iH}{\hbar}t}\, \widetilde\sigma^{N}(0)\, e^{\frac{iH}{\hbar}t},\quad \tr(\widetilde\sigma^{N}(t))=1.
$$
Hence, $\widetilde\sigma^{N}$ fulfills the hypothesis of theorem \ref{thm:1.3.2} and theorem \ref{mixed}. But  by using the triangular inequality
\begin{eqnarray*}
  \Vert \rho(t_{\rm rec}(\epsilon))-\rho(0)\Vert&\leq &\Vert \rho(t_{\rm rec}(\epsilon))-\sigma^{N}(t_{\rm rec}(\epsilon))\Vert\\
  & & +\Vert\sigma^{N}(t_{\rm rec}(\epsilon))-\sigma^{N}(0)\Vert+\Vert \sigma^{N}(0)-\rho(0)\Vert\\
  &\leq & 2\sqrt{\delta_N}+\Vert\sigma^{N}(t_{\rm rec}(\epsilon))-\sigma^{N}(0)\Vert
\end{eqnarray*} Then,

\begin{eqnarray*}
  \Vert \rho(t_{\rm rec}(\epsilon))-\rho(0)\Vert\leq  2\sqrt{\delta_N}+\Vert\widetilde\sigma^{N}(t_{\rm rec}(\epsilon))-\widetilde\sigma^{N}(0)\Vert \tr(\sigma^{N}(0))
\end{eqnarray*}
Observe that
$$
P_N:=\tr(\sigma^{N}(0))=\sum_{k=1}^N\rho_{k,k}(0)=\sum_{k=1}^N\braket{k\vert\rho(0)\vert k}=\sum_{k=1}^N\tr\left(\rho(0)\ket{k}\bra{k}\right)
$$
is the total probability of $\rho$ to be in one of the relevant $N$ states. By using inequalities (\ref{eq:8.5}) and (\ref{eq:10.5}) we can state the following two corollaries

\begin{corollary}\label{cor:1.1.3}
Let $\rho_0$ be a mixed state of the Hilbert space $\mathcal{H}$. Let $\rho(t)$ denote the unitary evolution given by the Hamiltonian $H$, \emph{i.e.},
$$
\dot\rho(t)=-\frac{i}{\hbar}[H,\rho(t)],\quad \rho(0)=\rho_0.
$$
Suppose that the spectrum of the Hamiltonian is discrete and $\rho_0$ has $N$ relevant states with error term $\delta_N$. Suppose that $N$-approximation to $\rho_0$, $\widetilde \sigma^{N}$ has non-zero uncertainty in the energy $\triangle E_{\widetilde \sigma^{N}}\neq 0$. Then, for any $\epsilon>0$ with
$$
\epsilon< \pi\cdot\min_{k\in\{1,\cdots,N\}}\left\{\sqrt{\tr\left(\widetilde \sigma^{N} \ket{k}\bra{k}\right)}
\right\}
$$where $\{\ket{k}\}$ is a basis of ortonormal eigenstates of the Hamiltonian, there exists a time $t_{\rm rec}(\epsilon)$ such that
\begin{equation}
\Vert \rho(t_{\rm rec}(\epsilon))-\rho(0)\Vert\leq  2\sqrt{\delta_N}+\sqrt{2}P_N\epsilon\sqrt{1-\frac{\epsilon^2}{8}},
\end{equation}
with
\begin{equation}
\epsilon\frac{\hbar}{\triangle E_{\widetilde\sigma^{N}}}\leq t_{\rm rec}(\epsilon)\leq \left(\frac{c_N}{\epsilon^{N-1}N^{\frac{N}{2}}}\right)\,\frac{\hbar}{\triangle E_{\widetilde\sigma^{N}}}
\end{equation}
where $c_N=(N-1)\Gamma(\frac{N-1}{2})4^{N-1}\pi^{\frac{N+1}{2}}$.
\end{corollary}

\begin{corollary}\label{cor:1.1.4}
Let $\rho_0$ be a mixed state of the Hilbert space $\mathcal{H}$. Let $\rho(t)$ denote the unitary evolution given by the Hamiltonian $H$, \emph{i.e.},
$$
\dot\rho(t)=-\frac{i}{\hbar}[H,\rho(t)],\quad \rho(0)=\rho_0.
$$
Suppose that the Hamiltonian has discrete spectrum and $\rho_0$ has $N$ relevant states with error term $\delta_N$.
Then, for any $t>0$ and any $\epsilon\in (0,1]$ there exists a time $t_{\rm rec}$ such that
\begin{equation}
\Vert \rho(t_{\rm rec}(\epsilon))-\rho(0)\Vert\leq  2\sqrt{\delta_N}+2P_N\left(1-\epsilon^2\right),
\end{equation}
with $t_{\rm rec}=j\cdot t$, $j\in \mathbb{N}$ and such that
$$
1\leq j\leq \sqrt{\pi}\frac{\Gamma(N^2)}{\Gamma(N^2+\frac{1}{2})}\frac{1}{\int_0^{\frac{\sqrt{2-2\epsilon}}{2}}\sin^{2N^2-2}(s)ds}\quad
.
$$
\end{corollary}

\section{Geometry of the mixed states space. Proof of Theorem \ref{mixed}}\label{sec5}
The most general state, the so-called \emph{mixed state}, is represented by a \emph{density operator}  in the Hilbert space $\mathcal{H}$. In this part of the paper we always assume that dim$(\mathcal{H})=n<\infty$, being $\mathcal{H}$ a vector space on the complex field ($\mathcal{H}=\mathbb{C}^n$). The density operator $\rho$ is in fact a \emph{density matrix}. Let us denote by $\mathcal{D}$ the space of density matrices. Recall that a density matrix is a complex matrix $\rho$  that  satisfies the following properties:
\begin{enumerate}
\item $\rho$ is a Hermitian matrix, i.e, the matrix coincides with its conjugate transpose matrix: $\rho=\rho^\dag$.
\item $\rho$ is positive $\rho \geq 0$. It means that any eigenvalue of $A$ is non-negative.
\item $\rho$ is normalized by the trace $\tr (\rho)=1$.
\end{enumerate}
Let us denote by
$$
\mathcal{P}^+:=\left\{\rho\in \mathcal{D}\,:\, \rho>0\right\}
$$
Let us consider the following sphere
$$
\widetilde S:=\{W\in M_n(\mathbb{C})\,:\,\tr(WW^\dag)=1\}
$$
and the following open dense set of $\widetilde S$,
$$
S:=\widetilde S\cap {\rm GL}(n,\mathbb{C})
$$
Since $M_n(\mathbb{C})$ is a vector space the tangent space $T_pM_n(\mathbb{C})$ at $p\in M_n(\mathbb{C})$ can be identified with $M_n(\mathbb{C})$ itself. Moreover we will denote by $g$ the Euclidean metric in $M_n(\mathbb{C})$, namely,
$$
g(X,Y)=\frac{1}{2}\tr(X^\dag Y+XY^\dag)
$$
and also we will denote by $g$ the restriction of the above metric tensor to $\widetilde{S}$. We are going to prove that the map $T_t:\widetilde S\to \widetilde S$ given by
$$
T_t(x)=e^{-\frac{i H t}{\hbar}}x
$$
is an isometry of $\widetilde S$. Suppose that we have two vectors $X,Y\in T_x\widetilde S$ then we need to check if
$$
g(X,Y)=g(dT_t(X),dT_t(Y))
$$
In order to do that consider the following two curves $\gamma_X:\mathbb{R}\to \widetilde S$ and $\gamma_Y:\mathbb{R}\to \widetilde S$ such that
\begin{eqnarray*}
\gamma_X(0)=\gamma_Y(0)=x,\quad \dot\gamma_X(0)=X,\quad \dot\gamma_Y(0)=Y.
\end{eqnarray*}

Then,
$$
dT_t(X)=\frac{d}{ds}T_t(\gamma_X(s))\vert_{s=0}=\frac{d}{ds}(e^{-\frac{i H t}{\hbar}}\gamma_X(s))\vert_{s=0}=e^{-\frac{i H t}{\hbar}}X
$$
For $dT_t(Y)$ we can obtain in an analogous way that $dT_t(Y)=e^{-\frac{i H t}{\hbar}}Y$. Hence,
\begin{eqnarray*}
g(dT_t(X),dT_t(Y))&=&g(e^{-\frac{i H t}{\hbar}}X,e^{-\frac{i H t}{\hbar}}Y)\\&=&\frac{1}{2}\tr ((e^{-\frac{i H t}{\hbar}}X)^\dag e^{-\frac{i H t}{\hbar}}Y+(e^{-\frac{i H t}{\hbar}}Y)^\dag e^{-\frac{i H t}{\hbar}}X)\\
&=&\frac{1}{2}\tr (X^\dag Y+Y^\dag X)=g(X,Y)
\end{eqnarray*}

That is what had to be proved. Since $T_t$ is an isometry in a metric space of finite measure and applying Theorem \ref{thm4} to $\widetilde S$, taking into account the volume of a geodesic ball in $\mathbb{S}^{2n^2-1}$ (see equation \ref{bola}) we conclude that
\begin{proposition}
For any $A\in \widetilde S$ and any $t\geq 0$, and any $r>0$ there exists $N_r$ such that
$$
{\rm d}^{\widetilde S}(A,T_{N_rt}(A))\leq r
$$
with
$$
1\leq N_r\leq \frac{\mu(\widetilde S)}{\mu(B_{r/2})}
$$
with $$
\frac{\mu(\widetilde S)}{\mu(B_{r/2})}=\frac{\int_0^{\pi}\sin^{2n^2-2}(s)ds}{\int_0^{r/2}\sin^{2n^2-2}(s)ds}
$$
\end{proposition}

Following Uhlmann \cite{Uhlmann86,Uhlmann87,Uhlmann89,Uhlmann91}, Bengtsson \cite{Bengtsson},  Chru\`sci\`nski \cite{Chruscinski}, and  D{\polhk{a}}browski \cite{Dabrowski1989,Dabrowski1990,Dabrowski1991} results, we can consider the following principal fiber bundle

$$
\begin{tikzcd}[column sep=2pc]
{\rm U}(n) \arrow{r}  &
  S\arrow{d}{\pi} \\
 &\mathcal{P}^+
\end{tikzcd}
$$
where the projection $\pi:S\to \mathcal{P}^+$ is given by $\pi(A)=AA^\dag$ and ${\rm U}(n)$ acts on $S$ by right multiplication, \emph{i.e.}, $(u,A)\to A u$ for $A\in S$ and $u\in{\rm U}(n)$. Taking into account that since $S$ is an open and dense subset of $\widetilde S$, we can endow $S$ with the restriction of the metric $g$ of $\widetilde S$.
Then, by using this metric structure the following fiber bundle
$$
\begin{tikzcd}[column sep=2pc]
{\rm U}(n) \arrow{r}  &
 (S,g)\arrow{d}{\pi} \\
 &(\mathcal{P}^+,g_B)
\end{tikzcd}
$$
becomes a Riemannian submersion. Where $g_B$ is the Bures metric in $\mathcal{P}^+$. With such a metric ${\rm U}(n)$ acts by isometries on $S$. Notice moreover that
$$
\frac{\partial}{\partial t}(\pi \circ T_t(A))=-\frac{i}{\hbar}\left[H,\pi\circ T_t(A)\right]
$$
Hence by using the global defined section $s:\mathcal{P}^+\to S$ given by
 $s(\rho)=\sqrt{\rho}$ (with a particular choice of the square root branch), the general solution of the von Neumann equation
$$
\frac{\partial}{\partial t}(\rho(t))=-\frac{i}{\hbar}\left[H,\rho(t)\right],\quad \rho(0)=\rho_0
$$
can be obtained as
\begin{equation}\label{seceq}
\rho(t)=\pi\left(T_t\left(s\left(\rho_0\right)\right)\right)
\end{equation}Taking into account that ${\rm d}^S={\rm d}^{\widetilde S}$ and that since $\pi$ is a Riemannian submersion, then
$$
{\rm d}_{\rm Bures}(\pi(A),\pi(T_{N_r t}(A)))\leq {\rm d}^S(A,T_{N_r t}(A))\leq r
$$
Hence, the theorem follows by using the above inequality for the particular case (see  equation (\ref{seceq})) of $A=s(\rho_0)$, because
$$
\sqrt{2-2F(\rho(N_r t),\rho_0)}={\rm d}_{\rm Bures}(\rho(t),\rho_0)\leq r
$$
and we can set
\begin{equation}\label{eqepsr}
\epsilon=1-\frac{r^2}{2}, \quad j=N_r,
\end{equation}
then
$$
F(\rho(j\cdot t),\rho_0)\geq \epsilon
$$

\section{Proof of theorem \ref{thm:1.3.2}}
The recurrence time for the Hamiltonian $H$ is the same that the recurrence time for the zero-point rescaled Hamiltonian $H_\lambda=H-\lambda I$. Given an initial state $\rho_0$, by using equation (\ref{seceq}), the temporal evolution is given by  $\rho(t)=\pi(e^{-\frac{it}{\hbar}(H-\lambda I)}W)$  where $W=s(\rho_0)$. But given the basis $\{\ket{k}\}$ of eigenvalues for the Hamiltonian,
$$
e^{-\frac{it}{\hbar}(H-\lambda I)}W=e^{-\frac{it}{\hbar}(H-\lambda I)}\sum_{k=1}^n\ket{k}\bra{k}W=\sum_{k=1}^ne^{-\frac{it}{\hbar}(E_k-\lambda)}\ket{k}\bra{k}W
$$
Then, the curve $\gamma(t)=e^{-\frac{it}{\hbar}(H-\lambda I)}W$ is a curve in the torus
$$
\mathbb{T}^n(W):=\left\{ \sum_{k=1}^ne^{i\theta_k}\ket{k}\bra{k}W\,:\, \theta_j\in[0,2\pi],\,\forall j \right\}
$$
We can make use of the following diffeomorphism
$$
\varphi:\mathbb{T}^n\to\mathbb{T}^n(W),\quad \varphi(e^{i\theta_1},\cdots,e^{i\theta_n})=\left(\sum_{k=1}^ne^{i\theta_k}\ket{k}\bra{k}W\right)
$$
and the inclusion map $\mathbb{T}^n(W)\subset M_{n}(\mathbb{C})$ to pull-back the metric from $M_n(\mathbb{C})$,\begin{eqnarray*}g(X_j,X_l)&=&\frac{1}{2}\tr\left(e_j^\dag e_l+e_l^\dag e_j\right)\\
  &=&\frac{1}{2}\tr\bigg(W^\dag\ket{j}\bra{j}e^{-i\theta_j}(-i)ie^{i\theta_l}\ket{l}\bra{l}W \\
  & & \quad\quad\quad+ W^\dag\ket{l}\bra{l}e^{-i\theta_l}(-i)ie^{i\theta_j}\ket{j}\bra{j}W\bigg)\\
  &=&\delta_{j,l}\cdot\tr\left(\rho_0 \ket{j}\bra{j}\right)
\end{eqnarray*}where $e_j=d\varphi(X_j)$ and  $\{X_1,\ldots,X_n\}$ is the basis of the Lie algebra $\mathfrak{t}^n$ (see \ref{torus}) given by
$$
X_j=\left(\overbrace{0,\ldots,0}^{j-1\, {\rm times}},ie^{i\theta_j},0,\ldots,0\right)
$$
The metric $g$ is a bi-invariant metric and since $[X_j, X_l]=0$, the torus $(\mathbb{T}^n,g)$ is a flat torus. In fact, $(\mathbb{T}^{n},g)$ geometrically is the following torus
$$
\mathbb{T}^{n}:=\overset{n}{\underset{k=1}{\times}}\mathbb{S}^1\left(\sqrt{\tr\left(\rho_0 \ket{k}\bra{k}\right)}\right)
$$
where $\mathbb{S}^1\left(\sqrt{\tr\left(\rho_0 \ket{k}\bra{k}\right)}\right)$ is the circle of radius $\sqrt{\tr\left(\rho_0 \ket{k}\bra{k}\right)}$. The injectivity radius (see proposition \ref{prop:8.6.2}) is given by
$
{\rm inj}(\mathbb{T}^n)=\pi\cdot\min_k\left\{\sqrt{\tr\left(\rho_0 \ket{k}\bra{k}\right)}
\right\}.
$
Moreover, the curve $\widetilde\gamma=\varphi^{-1}\circ\gamma$ is a geodesic curve because is the following curve
$$
\widetilde \gamma(t)=(e^{-i\frac{t}{\hbar}(E_1-\lambda)},\cdots,e^{-i\frac{t}{\hbar}(E_n-\lambda)})
$$
The length of $\widetilde\gamma([0,t])$ is given by

\begin{eqnarray}\label{eq:23.6.8}
{\rm length}{(\widetilde\gamma([0,t]))}&=&\int_0^t\sqrt{g(\dot{\widetilde\gamma}(s),\dot{\widetilde\gamma}(s))}ds\nonumber\\&=&\left\Vert-\frac{E_1-\lambda}{\hbar}X_1-\cdots-\frac{E_n-\lambda}{\hbar}X_n\right\Vert\, t\\
&=&\sqrt{\left(\frac{E_1-\lambda}{\hbar}\right)^2+\cdots+\left(\frac{E_n-\lambda}{\hbar}\right)^2}t.\nonumber
\end{eqnarray}

If $\sqrt{\left(\frac{E_1-\lambda}{\hbar}\right)^2+\cdots+\left(\frac{E_n-\lambda}{\hbar}\right)^2}\neq 0$, then for any $t<\frac{{\rm inj}(\mathbb{T}^2)}{\sqrt{\left(\frac{E_1-\lambda}{\hbar}\right)^2+\cdots+\left(\frac{E_n-\lambda}{\hbar}\right)^2}}$ we have
$$
{\rm dist}^{\mathbb{T}^n}(\widetilde\gamma(0),\widetilde\gamma(t))=\sqrt{\left(\frac{E_1-\lambda}{\hbar}\right)^2+\cdots+\left(\frac{E_n-\lambda}{\hbar}\right)^2}t.
$$
Hence by using the definition of  $t_{\rm rec}(\epsilon)$, for any $\epsilon<\frac{{\rm inj}(\mathbb{T}^2)}{\sqrt{\left(\frac{E_1-\lambda}{\hbar}\right)^2+\cdots+\left(\frac{E_n-\lambda}{\hbar}\right)^2}}$ we have
\begin{eqnarray}\label{eq:11}
  {\rm dist}^{\mathbb{T}^n}\left(\widetilde\gamma(0),\widetilde\gamma(t)\right)>\epsilon
\end{eqnarray} if $\frac{\epsilon}{\sqrt{\left(\frac{E_1-\lambda}{\hbar}\right)^2+\cdots+\left(\frac{E_n-\lambda}{\hbar}\right)^2}}<t<t_{\rm rec}(\epsilon)$.

Now, we are going to obtain upper bounds for the recurrence time by using the volume of the tube $\widetilde\gamma^{\theta}([0,t])$. Recall that the tube $\widetilde\gamma^{\theta}([0,t])$ is the set of points of $\mathbb{T}^n$ which  are at distance at most $\theta$ through normal geodesics emanating from $\widetilde\gamma([0,t])$. To estimate the volume of such a tube we first need to estimate the minimal focal distance of the tube.

The tube $\widetilde\gamma^{\theta}([0,t])$ with $\theta<\epsilon/2$ and $\epsilon<{\rm inj}(\mathbb{T}^n)$ has no self-intersections for $t<t_{\rm rec}(\epsilon)$. Because, otherwise suppose that $0\leq t_1<t_2<t_{\rm rec}(\epsilon)$, and $\Vert n_1\Vert=\Vert n_2\Vert=1$ in the normal bundle of $\widetilde\gamma$
$$
e^{in_1s_1}\widetilde\gamma(t_1)=e^{in_2s_2}\widetilde\gamma(t_2)
$$
Then $e^{in_1s_1}e^{\frac{-it_1}{\hbar}(H-\lambda I)}\widetilde\gamma(0)=e^{in_2s_2}\widetilde\gamma(t_2)$, but therefore
%% \begin{equation}\label{eq:12}\begin{aligned}
%% &e^{in_1s_1}e^{\frac{-it_1}{\hbar}(H-\lambda I)}\widetilde\gamma(0)=e^{in_2s_2}\widetilde\gamma(t_2),\quad e^{\frac{-it_1}{\hbar}(H-\lambda I)}e^{in_1s_1}\widetilde\gamma(0)=e^{in_2s_2}\widetilde\gamma(t_2),\\
%% &e^{\frac{-it_1}{\hbar}(H-\lambda I)}e^{in_1s_1}\widetilde\gamma(0)=e^{in_2s_2}\widetilde\gamma(t_2),\quad e^{in_1s_1}\widetilde\gamma(0)=e^{\frac{t_1}{\hbar}(H-\lambda I)}e^{in_2s_2}\widetilde\gamma(t_2)\\
%% &e^{i(n_1s_1-n_2s_2)}\widetilde\gamma(0)=e^{\frac{t_1-t_2}{\hbar}(H-\lambda I)}\widetilde\gamma(0)=\widetilde\gamma(t_2-t_1).
%% \end{aligned}
%% \end{equation}

\begin{eqnarray}\label{eq:12}
&e^{i(n_1s_1-n_2s_2)}\widetilde\gamma(0)=e^{\frac{t_1-t_2}{\hbar}(H-\lambda I)}\widetilde\gamma(0)=\widetilde\gamma(t_2-t_1)
\end{eqnarray}

 Since $\Vert i(n_1s_1-n_2s_2)\Vert\leq 2\theta<\epsilon$, and $\beta(t)=e^{i(n_1s_1-n_2s_2)t}\widetilde\gamma(0)$ is a geodesic curve joining $\beta(0)=\widetilde\gamma(0)$ and $\beta(1)=\widetilde\gamma(t_2-t_1)$,
$$
{\rm dist}^{\mathbb{T}^n}(\widetilde\gamma(0),\widetilde\gamma(t_2-t_1))<\epsilon
$$
if $t_2-t_1>\frac{\epsilon}{\sqrt{\left(\frac{E_1-\lambda}{\hbar}\right)^2+\cdots+\left(\frac{E_n-\lambda}{\hbar}\right)^2}}$ this is a contradiction with inequality (\ref{eq:11}). If otherwise $t_2-t_1\leq\frac{\epsilon}{\sqrt{\left(\frac{E_1-\lambda}{\hbar}\right)^2+\cdots+\left(\frac{E_n-\lambda}{\hbar}\right)^2}}$, then
$$
(t_2-t_1)\sqrt{\left(\frac{E_1-\lambda}{\hbar}\right)^2+\cdots+\left(\frac{E_n-\lambda}{\hbar}\right)^2}\leq\epsilon<{\rm inj}(\mathbb{T}^n).
$$
And hence, by (\ref{eq:12}),
$$
i\frac{t_1-t_2}{\hbar}(H-\lambda I)=i(n_1s_1-n_2s_2)
$$
But $i\frac{t_1-t_2}{\hbar}(H-\lambda I)$ belongs to the tangent bundle of $\widetilde\gamma$ and $i(n_1s_1-n_2s_2)$ belongs to the normal bundle of $\widetilde\gamma$, and hence a contradiction.

Since $\widetilde\gamma$ is a geodesic of $\mathbb{T}^n$, and since $\mathbb{T}^n$ is a flat manifold, there are no focal points along a normal geodesic to a geodesic of $\mathbb{T}^n$ (see \cite[proposition 2.12]{GimPal2015}), we have proved that there are no overlaps, then
$$
{\rm minfoc}(\widetilde\gamma([0,t]))<\theta.
$$
for any $t<t_{\rm rec}(\epsilon)$, $\theta<\epsilon/2$, $\epsilon<{\rm inj}(\mathbb{T}^n)=\pi\cdot\min_k\left\{\sqrt{\tr\left(\rho_0 \ket{k}\bra{k}\right)}\right\}$. Therefore, the $\theta$-tubular neighborhood $\widetilde\gamma^\theta([0,t])$ of $\widetilde\gamma([0,t])$ has volume  (see \cite[corollary 8.6]{tubes})
$$
\mu(\widetilde\gamma^\theta([0,t]))=\frac{2\pi^{\frac{n-1}{2}}}{(n-1)\Gamma(\frac{n-1}{2})}\theta^{n-1}\cdot {\rm length}(\widetilde\gamma([0,t]))
$$
But using equality (\ref{eq:23.6.8})
$$
\mu((\gamma^\theta([0,t]))=\frac{2\pi^{\frac{n-1}{2}}}{(n-1)\Gamma(\frac{n-1}{2})}\theta^{n-1}\cdot \frac{t}{\hbar}\sqrt{\left({E_1-\lambda}\right)^2+\cdots+\left({E_n-\lambda}\right)^2}
$$
Hence, taking into account that $\mu(\widetilde\gamma^\theta([0,t_{\rm rec}(\epsilon)])\leq\mu (\mathbb{T}^{n})$ we obtain
$$
t_{\rm rec}(\epsilon)\leq \hbar \frac{(n-1)\Gamma(\frac{n-1}{2})2^{n-1}\pi^{\frac{n+1}{2}}\underset{k}{\Pi}\sqrt{\tr\left(\rho_0 \ket{k}\bra{k}\right)}}{\theta^{n-1}\cdot \sqrt{\left({E_1-\lambda}\right)^2+\cdots+\left({E_n-\lambda}\right)^2}}
$$
setting $\lambda=\braket{H}_{\rho_0}=\tr\left(H\rho_0\right)$ we obtain
\begin{equation}\label{eq:14}
t_{\rm rec}(\epsilon)\leq \hbar C_n\frac{\underset{k}{\Pi}\sqrt{\tr\left(\rho_0 \ket{k}\bra{k}\right)}}{\theta^{n-1}\cdot \triangle E_{\rho_0}
}
\end{equation}
where $\triangle E_{\rho_0}=\sqrt{\tr\left(H^2\rho_0\right)-\tr\left(H\rho_0\right)^2}$ and $C_n=(n-1)\Gamma(\frac{n-1}{2})2^{n-1}\pi^{\frac{n+1}{2}}$.
Since $\theta$ is a distance in $\mathbb{T}^{n}(W)$, $\mathbb{T}^{n}(W)\subset S$ and $\pi:S\to\mathcal{P}^+$ is a Riemannian submersion , then
$$
{\rm d}_{\rm Bures}\left(\rho_o,\rho(t_{\rm rec})\right)\leq {\rm d}^S\left(\gamma(0),\gamma(t_{\rm rec})\right)\leq \theta
$$
therefore
$$
\sqrt{2-2F(\rho(t_{\rm rec}),\rho_0)}={\rm d}_{\rm Bures}(\rho(t),\rho_0)\leq \theta<\frac{\epsilon}{2}.
$$
%%\textcolor{red}{repasar}
But observe that inequality (\ref{eq:14}) holds for any $\theta<\epsilon/2$, then letting $\theta$ tend to $\epsilon/2$ we obtain
$$
t_{\rm rec}(\epsilon)\leq \hbar c_n\frac{\underset{k}{\Pi}\sqrt{\tr\left(\rho_0 \ket{k}\bra{k}\right)}}{\epsilon^{n-1}\cdot \triangle E_{\rho_0}
}
$$
with $c_n=(n-1)\Gamma(\frac{n-1}{2})4^{n-1}\pi^{\frac{n+1}{2}}$.

\section{Discussion}

In this paper we have found upper bounds for the recurrence time of a quantum mixed state with discrete spectrum of energies. By the case of discrete but finite spectrum, we have obtained two type of upper bounds; one of them depends on the uncertainty in the energy, and the other depends only on the  \emph{number of (relevant) states} based on the decomposition in the basis of eigenstates of the Hamiltonian. On the other hand, in the case of discrete but non-finite spectrum, using the same reasoning, we obtained two upper bounds defining the number of relevant states according to an statistical measurement.

To obtain this bounds, we estimate the volume of a geodesic ball  in a sphere or the volume of a tube around a geodesic curve in a torus. The most general result related with uncertainty in the energy can be read as
\begin{equation}\label{eq20:v7}
\epsilon\frac{\hbar}{\triangle E_{\widetilde\sigma^{N}}}\leq t_{\rm rec}(\epsilon)\leq \left(\frac{c_N}{\epsilon^{N-1}N^{\frac{N}{2}}}\right)\,\frac{\hbar}{\triangle E_{\widetilde\sigma^{N}}}
\end{equation}
for a mixed state $\rho$ with $N$ relevant states. That means that the mixed state $\rho$ admits a finite dimensional approximation $\sigma^{N}$ satisfying $\Vert \rho-\frac{1}{\tr(\sigma^{N})}\sigma^{N}\Vert^2=\delta_N$ with $\delta_N\to 0$ when $N\to\infty$ and $P_N=\tr(\sigma^{N})\to 1$ when $N\to\infty$. In fact, this $\delta_N$ gives the accuracy in the estimation of the Poincar\'e recurrence time because
$$
\Vert \rho(t_{\rm rec}(\epsilon))-\rho(0)\Vert\leq  2\sqrt{\delta_N}+\sqrt{2}P_N\epsilon\sqrt{1-\frac{\epsilon^2}{8}}.
$$
Moreover, we can state and prove that upper bounds for the recurrence time can be obtained without using the specific Hamiltonian (or uncertainty in the energy) but only the number of relevant states because for any $t>0$ and any $\epsilon\in (0,1]$ there exists \begin{equation}\label{eq21:v7}t\leq t_{\rm rec}(\epsilon)\leq C(N)t\end{equation} with $C(N)$ depending only on $N$, such that
  $$
\Vert \rho(t_{\rm rec}(\epsilon))-\rho(0)\Vert\leq  2\sqrt{\delta_N}+2P_N\left(1-\epsilon^2\right).
  $$
In the proof of the lower bound for the recurrence time it is established the existence of a  quantum speed limit of a system $t_{QSL} \simeq \frac{\hbar}{\triangle E}$ \cite{Mandelstam,Levitin}. This is a pure quantum phenomenon, as it depends explicitly on the constant $\hbar$. The energy-time uncertainty relation for time-independent systems has been further extended in \cite{Giovanetti2003,Giovanetti2004}, who determined the quantum speed limit time for general mixed states, no necessary orthogonal, as a function of their geometrical angle given by the Bures length. This quantum speed limit time is reduced when the Bures length between the initial an final states is smaller. Whereas  $t_{QSL}$ is determined by the initial mean energy in time-independence case or energy variance in the case of driven quantum systems \cite{Deffner2013,Zhang2014}.
In the classical limit, $\hbar \to 0$, the quantum speed limit disappeared and, in principle, progresses arbitrarily fast.

Someone could expect that in the classical limit the quantum Poincar\'e recurrence becomes the classical one. But it should be noticed that, similarly to what happens with the quantum speed limit, the classic limit  for the  the Poincar\'e recurrence time related with the uncertainty in the energy is zero.  Since by taking the limit $\hbar\to 0$ in inequality (\ref{eq20:v7}),
$$
t_{\rm rec}(\epsilon)\to 0,\quad{\rm as}\quad \hbar\to 0.
$$
This shows that this recurrence is a purely quantum phenomena. Nevertheless, since inequality (\ref{eq21:v7}) depends only in the number of relevant states $N$, it remains unaltered by the limit $\hbar \to 0$. Hence this phenomena survives to the classical limit. Our conjecture is that the classical Poincar\'e recurrence time is related with the number of relevant states (as a measure of the volume of the classical phase space).  Finally, in view of that we conclude that the inequalities (\ref{eq20:v7}) and (\ref{eq21:v7}) are not two bounds for the same quantum recurrence process but two bounds for two different phenomena. The first one is purely quantum and the second one is related with the classical recurrence principle.

%Esto ha sido ampliamente estudiado y tiene implicaciones importantes in quantum computation and information processing devices.

\ack
The research of the first author (VG) has been supported in part by Universitat Jaume I Research Program Project P1-1B2012-
18, and DGI-MINECO grant (FEDER) MTM2013-48371-C2-2-PDGI from Ministerio de Ciencia
e Innovaci\'on (MCINN), Spain. The research of the second author (JMS) was supported in part by DGI-MINECO grant (FEDER) MTM2013-48371-C2-2-PDGI from Ministerio de Ciencia
e Innovaci\'on (MCINN), Spain.

\appendix
\section*{Appendix}
\setcounter{section}{1}

\subsection{Poincar\'e recurrence for isometries in metric spaces of finite measure}To prove theorem \ref{mixed}, we have used the following Poincar\'e recurrence theorem in metric spaces

\begin{theorem}\label{thm4}
  Let $(M,d,\mu)$ be a metric space $(M,d)$ with finite measure $\mu(M)<\infty$. Then for  any volume preserving isometry $T:M\to M$, any point $p\in M$, and any $r>0$ there exists $N_r\in \mathbb{N}$ such that the distance $d(p,T^{N_r}(p))$ from $p$ to $T^{N_r}(p)=\overbrace{T\circ\cdots\circ T}^{N_r\, {\rm times }}(p)$ is bounded from above by $r$, ($d(p,T^{N_r}(p))<r$), with $N_r$ satisfying the following inequality
$$
1\leq N_r\leq \frac{\mu(M)}{\mu(B_{r/2}(p))}
$$where $B_{r/2}(p)$ denotes the metric ball centered at $p\in M$ of radius $r/2$. Namely,
$$
B_{r/2}(p):=\left\{x\in M\, :\, d(p,x)<\frac{r}{2}\right\}
$$
\end{theorem}
\begin{proof}
To prove that theorem we first need the following lemma
\begin{lemma}There exists $N_r\in \mathbb{N}$ with
$
1\leq N_r\leq \frac{\mu(M)}{\mu(B_{r/2}(p))}
$
such that
$$
T^{N_r}\left(B_{r/2}(p)\right)\cap B_{r/2}(p)\neq \emptyset
$$
\end{lemma}
\begin{proof}
Suppose that there exists $N\geq 1$ such that
$
 T^i\left(B_{r/2}(p)\right)\cap T^j\left(B_{r/2}(p)\right)=\emptyset
$
for all $i,j<N$, $i\neq j$, then
$$
\mu\left(\overset{N}{\underset{i=1}{\cup}}T^i\left(B_{r/2}(p)\right)\right)=\sum_{i=1}^N\mu\left(T^i\left(B_{r/2}(p)\right)\right)\leq \mu(M)
$$ But since $T$ is a volume-preserving transformation,
$
N\mu\left(B_{r/2}(p)\right)\leq \mu(M).
$
If we take $S$ as the first integer such that
$$
S>\frac{\mu(M)}{\mu\left(B_{r/2}(p)\right)}, \quad \left(S=\left\lceil \frac{\mu(M)}{\mu\left(B_{r/2}(p)\right)}\right\rceil\right).
$$
then, there must exists $0<1\leq i<j\leq S$ such that
$
T^i(B_{r/2})\cap T^j(B_{r/2})\neq 0
$
, but taking $T^{-i}$ in this expression we obtain,
$$
B_{r/2}\cap T^{j-i}(B_{r/2})\neq 0
$$
the lemma follows if we set $N_r=j-i$ and taking into account that
$
1\leq j-i\leq S-1
$
and that
$
S\leq \frac{\mu(M)}{\mu\left(B_{r/2}(p)\right)}+1
$
.
\end{proof}
Applying the above lemma there exists $q\in B_{r/2}$ such that
$
T^{N_r}(q)\in B_{r/2}(p)
$. This implies that
$
{\rm d}(p,T^{N_r}(q))\leq \frac{r}{2}
$
but since $T$ is an isometry
$
{\rm d}(T^{N_r}(p),T^{N_r}(q))={\rm d}(p,q)
$
and hence, by the triangular inequality,
$$
{\rm d}(p,T^{N_r}(p))\leq {\rm d}(p,T^{N_r}(q))+{\rm d}(T^{N_r}(p),T^{N_r}(q))\leq \frac{r}{2} +\frac{r}{2}=r
$$
\end{proof}

\subsection{Volume of balls in real space forms}\label{sec7.1}
Let $\mathbb{M}_\kappa^n$ be the simply connected Riemannian manifold of constant sectional curvature $\kappa$ and dimension $n$. About each point $p\in \mathbb{M}_\kappa^n$ there exists a coordinate system $(t,\theta)\in [0,\pi/\sqrt{\kappa}]\times \mathbb{S}^{n-1}$, relative to which the Riemannian metric reads as (see \cite[pag 39]{Chavel})
$$
ds^2=(dt)^2+{\rm S}_\kappa^2(t)\vert d\theta\vert^2
$$
where ${\rm S}_\kappa(t)$ is the solution to the following differential equation with initial conditions
$$
{\rm S}_\kappa''+\kappa{\rm S}_\kappa=0,\quad {\rm S}_\kappa(0)=0,\quad{\rm S}_\kappa'(0)=1.
$$
Observe that the $t$-curves are geodesics and in the particular case of spheres (spaces of $\kappa=1$),
$$
{\rm S}_1(t)=\sin(t)
$$
then, the Riemannian volume element is
$$
d{\rm V}=\sin^{n-1}(t)d{\rm V}_{\mathbb{S}^{n-1}}
$$
where $dV_{\mathbb{S}^{n-1}}$ the Riemannian volume element in $\mathbb{S}^{n-1}$. The volume of the geodesic ball $B_r$ of radius $r$ in ${\mathbb{S}^{n}_1}$ can be obtained as
\begin{equation}\label{bola}
{\rm V}(B_r)={\rm V}(\mathbb{S}^{n-1})\int_0^r\sin^{n-1}(t)dt=\frac{2\pi^{n/2}}{\Gamma(n/2)}\int_0^r\sin^{n-1}(t)dt
\end{equation}
See also \cite[pag 252]{tubes} for the general expression of the volume of a geodesic ball.

\subsection{The injectivity radius of a $n$-dimensional flat torus}\label{torus}

Let $\displaystyle\mathbb{T}^n=\overbrace{U(1)\times\ldots\times U(1)}^{n\text{-times}}$ denote the $n$-dimensional torus, namely,
$$
\mathbb{T}^n=\left\{\left(e^{i\theta_1},\ldots,e^{i\theta_n}\right)\,:\, \theta_1,\cdots,\theta_n\in\mathbb{R}\right\}
$$
with the usual product law
$$
\left(e^{i\theta_1},\ldots,e^{i\theta_n}\right)\star \left(e^{i\alpha_1},\ldots,e^{i\alpha_n}\right)=\left(e^{i(\theta_1+\alpha_1)},\ldots,e^{i(\theta_n+\alpha_n)}\right).
$$
Taking derivatives at $t=0$ in
$
\left(e^{i\theta_1},\ldots,e^{i\theta_n}\right)\star \left(e^{it},1,\ldots,1\right)
$
we get the left invariant vector field
$$
X_1=\left(ie^{i\theta_1},0,\ldots,0\right),
$$
similarly,
$$
X_2=\left(0,ie^{i\theta_2},0,\ldots,0\right),\quad\cdots\quad X_n=\left(0,0\cdots,0,ie^{i\theta_n}\right)
$$
we obtain a basis $\{X_1,\cdots,X_n\}$ of the Lie algebra $\mathfrak{t}^n$ of $\mathbb{T}^n$. In fact, since the group is abelian, $\{X_1,\cdots X_n\} $ are right invariant vector fields as well. Now given a $n$-tuple of non-zero  real numbers $(g_1,\dots,g_n)\in \mathbb{R}^n$ we can define the bi-invariant metric
$$
g(X_j,X_k)=\delta_{j,k}\, g_j^2
$$
The sectional curvature $\kappa(X_j,X_k)$ of the plane spanned by $X_j$ and $X_k$ in a Lie group with bi-invariant metric is given by (see \cite{Milnor-lie}) $\kappa(X_j,X_k)=\frac{1}{4}\Vert [X_j,X_k]\Vert^2$. But, as can be easily checked, $[X_j,X_k]=0$, and therefore $(\mathbb{T}^n,g)$ is a flat torus. Moreover, $(\mathbb{T}^n,g)$ as a Riemannian manifold is isometric to
$$
\mathbb{R}^n/\Lambda=\left(\mathbb{S}^1,g_1^2d\theta_1^2\right)\times\cdots\times \left(\mathbb{S}^1,g_n^2d\theta_n^2\right)
$$
because
$$
\varphi: \mathbb{R}^n/\Lambda\to (\mathbb{T}^n,g) ,\quad (\theta_1,\ldots,\theta_n)\to\varphi(\theta_1,\ldots,\theta_n)=\left(e^{i\theta_1},\ldots,e^{i\theta_n}\right)
$$
is a Riemannian isometry since
$$
d\varphi(v_1\partial \theta_1,\ldots v_n\partial \theta_n)=v_1X_1+\cdots+v_nX_n,\quad{\rm and}\quad ds^2(u,v)=g(d\varphi(u),d\varphi(v)).
$$
Then we can obtain the volume of $(\mathbb{T}^n,g)$ as
$$
{\rm vol}(\mathbb{T}^n,g)=\left(2\pi\right)^n\overset{n}{\underset{j=1}{\Pi}}\vert g_j\vert
$$
and the geodesic curves starting at $p=\left(e^{i\theta_1},\ldots,e^{i\theta_n}\right)$ with tangent vector $v\in T_{p}\mathbb{T}^n$, $v=v_1X_1(p)+\cdots v_nX_n(p)$ are given by (see \cite[corollary 57]{Oneill} )
$$
\gamma(t)=\left(e^{iv_1t},\ldots,e^{iv_nt}\right)\star\left(e^{i\theta_1},\ldots,e^{i\theta_n}\right)
$$
moreover we can obtain the injectivity radius
\begin{proposition}\label{prop:8.6.2}
The injectivity radius ${\rm inj}(\mathbb{T}^n,g)$ of $(\mathbb{T}^n,g)$ is given by
$$
{\rm inj}(\mathbb{T}^n,g)=\pi\min_{j\in\{1,\ldots,n\}}\left\{\vert g_j\vert\right\}
$$
\end{proposition}
\begin{proof}
Since $(\mathbb{T}^n,g)$ is a flat manifold the injectivity radius is given (see \cite[corollary 4.14 of chap. III]{Sakai}) by the one half of the length of the shortest closed non-trivial geodesic. But since $(\mathbb{T}^n,g)$ is isometric to $\mathbb{R}^n/\Lambda$, and $\mathbb{R}^n/\Lambda$ is a product manifold, this length is the length of the shortest closed geodesic of one of the factors (see \cite[corollary 57]{Oneill}).
\end{proof}

%% \subsection{Distances through a Riemannian submanifold}\label{submersion}
%% \begin{proposition}Given a Riemannian submerion
%% $$
%% \begin{tikzcd}[column sep=2pc]
%% F \arrow{r}  &
%%  (M,g)\arrow{d}{\pi} \\
%%  &(B,g_B)
%% \end{tikzcd}
%% $$
%% then $\pi$ is a non expansive map, \emph{i.e.},
%% $$
%% {\rm d}_B(\pi(p),\pi(q))\leq {\rm d}_M(p,q.)
%% $$
%% \end{proposition}

%% \begin{acknowledgements} One of the reasons for having written
%% \end{acknowledgements}
\section*{References}

%% \bibliographystyle{iopart-num}
%% \bibliography{../../../tesis}

\def\cprime{$'$} \def\polhk#1{\setbox0=\hbox{#1}{\ooalign{\hidewidth
  \lower1.5ex\hbox{`}\hidewidth\crcr\unhbox0}}}
  \def\polhk#1{\setbox0=\hbox{#1}{\ooalign{\hidewidth
  \lower1.5ex\hbox{`}\hidewidth\crcr\unhbox0}}}
  \def\polhk#1{\setbox0=\hbox{#1}{\ooalign{\hidewidth
  \lower1.5ex\hbox{`}\hidewidth\crcr\unhbox0}}} \def\cprime{$'$}
  \def\cprime{$'$} \def\cprime{$'$} \def\cprime{$'$} \def\cprime{$'$}
\providecommand{\newblock}{}

\end{document}